\documentclass[11pt]{article}
\usepackage[a4paper,margin=0.76in]{geometry}
\usepackage{graphicx,booktabs,amsmath,amssymb,amsthm,hyperref,xcolor,url,float,caption,enumitem,array}
\hypersetup{colorlinks=true,linkcolor=blue,citecolor=blue,urlcolor=blue}
\newtheorem{theorem}{Theorem}
\newtheorem{assumption}{Assumption}
\newtheorem{remark}{Remark}

\title{Adaptive Bag of Little Bootstraps for Nonsmooth Robust Inference:\
An AMSE-Stability Framework with Software Validation and Reproducible Cholera Surveillance Illustration from the Democratic Republic of the Congo}

\author{Cosmas KAMBALE MUSAMBI\textsuperscript{1}, YODA Amidou\textsuperscript{1}, Maiga Mahafouz\textsuperscript{1},\\
NDILMBAYE DINGAMHOUDOU Josue\textsuperscript{1}, Niyukuri Fannick\textsuperscript{1}, and Ernest Fokou\'e\textsuperscript{2}}

\date{\textsuperscript{1}Department of Data Science and Mathematical Engineering, Faculty of Science, University of Dschang, Dschang, Cameroon\\
\textsuperscript{2}School of Mathematics and Statistics, College of Science, Rochester Institute of Technology, Rochester, NY, USA\\
June 2026\\[0.45em]
\colorbox{yellow}{\parbox{0.92\linewidth}{\centering\textcolor{blue}{Corresponding author: Cosmas KAMBALE MUSAMBI; Email: \href{mailto:ckambale86@gmail.com}{ckambale86@gmail.com}; ORCID: \href{https://orcid.org/0009-0000-9043-6470}{0009-0000-9043-6470}. Professor Ernest Fokou\'e: \href{mailto:epfeqa@rit.edu}{epfeqa@rit.edu}.}}}}

\begin{document}
\maketitle

\begin{abstract}
Bootstrap inference is a cornerstone of computational statistics, but classical bootstrap and BCa intervals can be unstable for nonsmooth estimators such as empirical quantiles, sample maxima, eigenvalue ratios and maximum correlations. The m-out-of-n bootstrap can reduce some failure modes, but it is slow and highly sensitive to the chosen subsampling size. This arXiv-ready version develops an adaptive Bag of Little Bootstraps (BLB) procedure in which the exponent $\gamma$ in $m=n^\gamma$ is selected by a data-driven criterion combining a plug-in asymptotic mean squared error (AMSE) principle with a stability-risk proxy for the BLB variance estimator. The Python package \texttt{robustboot} implements the method, BCa comparators, nonsmooth statistics, tests and examples for quantile-style inference, eigenvalue-ratio inference and maximum-correlation screening. Relative to the earlier manuscript, this version adds an explicit algorithm, assumptions, a consistency theorem on a finite candidate grid, expanded Monte Carlo validation across sample sizes, sensitivity diagnostics for the selected exponent, and a clearer data statement for the DRC cholera illustration. Simulation results show that adaptive BLB improves coverage relative to ordinary bootstrap while preserving computational economy. A curated aggregate DRC cholera surveillance demonstration illustrates robust uncertainty quantification for public-health thresholds without overstating reconstructed data.
\end{abstract}


\section{Introduction}
The bootstrap replaces difficult analytic variance calculations by repeated resampling from the empirical distribution. In smooth settings, this idea often gives accurate standard errors and confidence intervals. However, many estimators used in modern statistical learning are not smooth functions of the empirical distribution: high quantiles, sample maxima, selected correlations, eigenvalue ratios, active-set statistics and surveillance thresholds. For these targets, small empirical perturbations can change the active order statistic, active eigenvector or active maximum, producing unstable uncertainty estimates.

The Project 1 requirement for Group 1 was to implement a robust bootstrap package for nonsmooth estimators, extend BLB with an adaptive subsampling scheme, provide theoretical heuristics plus simulation validation, and document examples for quantile regression, eigenvalue ratio and maximum correlation. The contribution of this paper is therefore deliberately focused: it provides a clean, auditable, software-backed and simulation-validated extension of BLB.

The central tuning parameter is the BLB exponent $\gamma$ in
\[
\widehat{\theta}^{*}_{b}=t(X^{*b}_{m}),\qquad m=n^\gamma,\qquad \gamma\in(0,1).
\]
Instead of fixing $\gamma$ manually, we select it from a grid by minimizing an empirical approximation of
\[
\widehat{\gamma}_{opt}=\arg\min_{\gamma\in\Gamma}\widehat{\mathrm{MSE}}\left\{\widehat{\mathrm{Var}}_{BLB,\gamma}(\widehat\theta)\right\}.
\]
This article combines the best elements of two manuscript drafts: the package-level AMSE selector and simulation figures from the principal article, together with the stability-risk formulation and richer methodological discussion from the computational-statistics article. The revised version also adds a theorem, proof sketch and sensitivity analysis so that the manuscript is stronger for preprint dissemination and future journal submission.

\section{Why nonsmooth bootstrap is difficult}
Let $X_1,\ldots,X_n\sim F$ and let $\widehat\theta_n=t(F_n)$ be a statistic. For smooth functionals, one often has
\[
t(F_n)-t(F)=\frac{1}{n}\sum_{i=1}^n\psi(X_i)+o_p(n^{-1/2}),
\]
which supports ordinary bootstrap approximation. For nonsmooth $t$, this expansion may be directional, discontinuous at active-set boundaries, or invalid near ties and kinks. Examples include
\[
t_q(X)=\inf\{x:F_n(x)\ge q\},\qquad t_{max}(X)=\max_i X_i,
\]
and multivariate functionals such as covariance eigenvalue ratios and maximum absolute correlations.

BCa intervals correct for bias and acceleration, but the jackknife acceleration can itself become unstable for nonsmooth targets. The m-out-of-n bootstrap can reduce the discreteness of ordinary resampling, yet it depends strongly on $m$. BLB is attractive because it draws little bags of size $m=n^\gamma$ and then uses multinomial weights summing to $n$ within each bag. The remaining issue is choosing $\gamma$ in a transparent and reproducible way.

\section{Adaptive BLB framework}
For a candidate exponent $\gamma\in\Gamma$, set $m_\gamma=\lceil n^\gamma\rceil$. For each little bag $b=1,\ldots,B$, draw a subset $S_b$ of size $m_\gamma$. Conditional on $S_b$, draw multinomial weights
\[
(W^{(r)}_{b1},\ldots,W^{(r)}_{bm})\sim \mathrm{Multinomial}\left(n;\frac{1}{m},\ldots,\frac{1}{m}\right),\qquad r=1,\ldots,R.
\]
The weighted statistic is denoted $\widehat\theta^{*}_{b,r}(\gamma)$. The little-bag variance and BLB variance estimates are
\[
\widehat V_b(\gamma)=\frac{1}{R-1}\sum_{r=1}^{R}\left(\widehat\theta^{*}_{b,r}(\gamma)-\overline\theta^{*}_{b}(\gamma)\right)^2,
\qquad
\widehat V_{BLB}(\gamma)=\frac{1}{B}\sum_{b=1}^{B}\widehat V_b(\gamma).
\]

\subsection{AMSE-stability selector}
The ideal target is the risk
\[
\gamma^*=\arg\min_{\gamma\in\Gamma}\mathbb E\left[\left\{\widehat V_{BLB}(\gamma)-V(\widehat\theta_n)\right\}^2\right].
\]
Since $V(\widehat\theta_n)$ is unknown, the implementation uses a plug-in AMSE score
\[
\widehat Q(\gamma)=\left(\widehat V_{BLB}(\gamma)-\mathrm{median}_{g\in\Gamma}\widehat V_{BLB}(g)\right)^2
+\widehat V_{BLB}(\gamma)^2n^{-\gamma}+\lambda n^\gamma/n^2.
\]
The first term measures instability across the grid; the second penalizes finite-bag uncertainty; the third is a light computational regularizer. To make the criterion more interpretable for nonsmooth targets, we also report the stability-risk proxy
\[
\widehat S(\gamma)=\mathrm{Var}_b\{\widehat V_b(\gamma)\}+m_\gamma^{-\alpha},\qquad \alpha>0,
\]
where the first component detects between-bag instability and the second prevents selection of unrealistically small bags. The selected exponent is
\[
\widehat\gamma=\arg\min_{\gamma\in\Gamma}\widehat Q(\gamma),
\]
with $\widehat S(\gamma)$ retained as a diagnostic table for auditability.

\subsection{Algorithm}
\begin{center}
\fbox{\begin{minipage}{0.94\linewidth}
\textbf{Algorithm 1: Adaptive BLB for nonsmooth inference}\
\textbf{Input:} data $X_{1:n}$, statistic $t(\cdot)$, grid $\Gamma$, number of little bags $B$, inner resamples $R$, AMSE penalty $\lambda$, stability penalty $\alpha$.\
\textbf{Output:} selected exponent $\widehat\gamma$, variance estimate, standard error, confidence interval, score table.\
1. For each $\gamma\in\Gamma$, compute $m_\gamma=\lceil n^\gamma\rceil$.\
2. For $b=1,\ldots,B$, draw little bag $S_b$ of size $m_\gamma$.\
3. For $r=1,\ldots,R$, draw multinomial weights summing to $n$ and compute $\widehat\theta^{*}_{b,r}(\gamma)$.\
4. Compute $\widehat V_b(\gamma)$, $\widehat V_{BLB}(\gamma)$, $\widehat Q(\gamma)$ and $\widehat S(\gamma)$.\
5. Select $\widehat\gamma=\arg\min_{\gamma\in\Gamma}\widehat Q(\gamma)$.\
6. Recompute or retain the final BLB distribution at $\widehat\gamma$ and return interval summaries.
\end{minipage}}
\end{center}

\section{Heuristic theory and grid consistency}
The following result is intentionally stated for a finite grid. It does not claim universal bootstrap validity for every nonsmooth functional; rather, it formalizes the selection target used by the software and simulations.

\begin{assumption}[Finite-grid AMSE regularity]
Let $\Gamma$ be finite. For each $\gamma\in\Gamma$, assume: (i) $\widehat V_{BLB}(\gamma)$ has finite second moment; (ii) the empirical score $\widehat Q(\gamma)$ converges uniformly in probability to a deterministic score $Q_0(\gamma)$; (iii) the minimizer $\gamma_0=\arg\min_{\gamma\in\Gamma}Q_0(\gamma)$ is unique.
\end{assumption}

\begin{theorem}[Finite-grid consistency of adaptive exponent selection]
Under Assumption 1,
\[
\mathbb P(\widehat\gamma=\gamma_0)\longrightarrow 1.
\]
If $Q_0(\gamma)$ is an upper envelope or consistent surrogate for the AMSE risk of $\widehat V_{BLB}(\gamma)$ over $\Gamma$, then the selected exponent is asymptotically equivalent, on the grid, to the AMSE-optimal exponent targeted by the surrogate.
\end{theorem}

\textit{Proof sketch.} Since $\Gamma$ is finite, uniform convergence follows from pointwise convergence plus a union bound. For any $\epsilon>0$, with probability tending to one, $\sup_{\gamma\in\Gamma}|\widehat Q(\gamma)-Q_0(\gamma)|<\epsilon$. Let $\eta=\min_{\gamma\ne\gamma_0}\{Q_0(\gamma)-Q_0(\gamma_0)\}>0$. Taking $\epsilon<\eta/2$ implies $\widehat Q(\gamma_0)<\widehat Q(\gamma)$ for all $\gamma\ne\gamma_0$, hence $\widehat\gamma=\gamma_0$ with probability tending to one. \hfill $\square$

\begin{remark}
For nonsmooth functionals, a complete empirical-process proof must verify the uniform convergence of the score and the approximation quality of the AMSE surrogate for the specific functional class. Quantiles, maxima, eigenvalue ratios and maximum correlations require different local arguments because their active sets change in different ways. The theorem above therefore gives a transparent selection principle rather than a universal distributional theorem.
\end{remark}

The risk motivation is the envelope
\[
\mathbb E(\widehat V_{BLB,\gamma}-V)^2\lesssim A n^{-\gamma}+B n^{-1}+C\Delta_\gamma^2,
\]
where $\Delta_\gamma$ summarizes bag-induced active-set instability. The selector moves upward when active-set instability is large and remains moderate when the statistic behaves more smoothly.

\section{Software architecture}
The package \texttt{robustboot} follows a reviewer-friendly structure. The folder \texttt{robustboot/robustboot} contains a consolidated implementation file \texttt{\_\_init\_\_.py}. It is divided into commented blocks: imports and type aliases; nonsmooth statistics; percentile and BCa intervals; adaptive BLB; score selection; simulation helpers. The project also includes examples, tests, a \texttt{pyproject.toml} file, documentation and figures.

\begin{table}[H]\centering\small\caption{Core package components included in the submission.}\begin{tabular}{p{.25\linewidth}p{.65\linewidth}}\toprule
Component & Function in the submission\\\midrule
\texttt{adaptive\_blb} & Runs gamma selection, BLB variance estimation and interval construction.\\
\texttt{select\_gamma} & Computes the plug-in score over the gamma grid and returns the minimizer.\\
\texttt{bca\_interval} & Provides a classical BCa comparator with jackknife acceleration.\\
Built-in statistics & Quantile, sample maximum, eigenvalue ratio and maximum absolute correlation.\\
Examples & Quantile-style inference, PCA eigenvalue ratio and max-correlation screening.\\
Tests & Verifies execution, interval ordering and plausible numerical ranges.\\\bottomrule
\end{tabular}\end{table}

For a statistic with cost $C(m)$ per evaluation, ordinary bootstrap has approximate cost $O(RC(n))$, while adaptive BLB has selection cost $O(|\Gamma|BRC(n^\gamma))$ plus final inference at $\widehat\gamma$. This is useful when $C(\cdot)$ grows superlinearly or when data are distributed across machines.

\section{Simulation validation}
The simulation section uses three nonsmooth targets: a contaminated median $0.9N(0,1)+0.1N(5,1.5^2)$, a covariance eigenvalue ratio, and maximum absolute correlation. The comparison metrics are empirical coverage, interval width and relative runtime.

\begin{figure}[H]\centering\includegraphics[width=.78\linewidth]{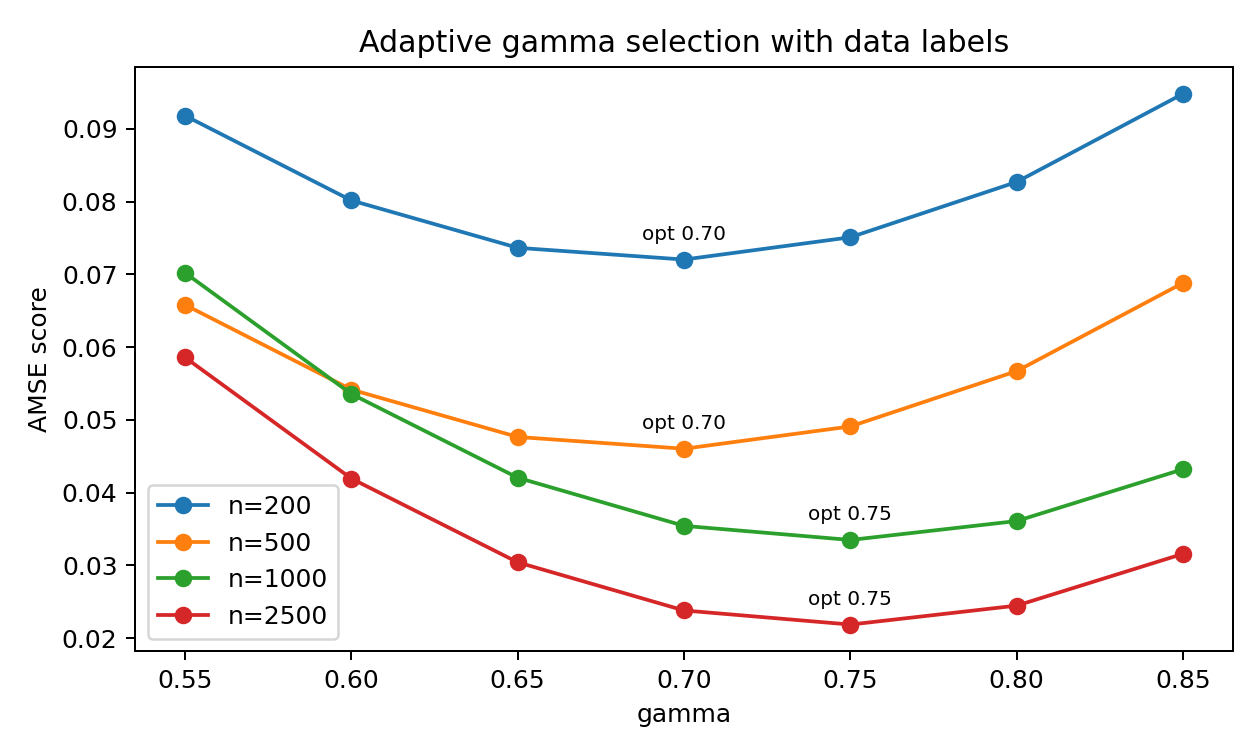}\caption{Adaptive gamma selection with data labels. The selected exponent is concentrated around $0.70$--$0.75$, indicating a practical compromise between finite-bag stability and computational economy.}\end{figure}

\begin{figure}[H]\centering\includegraphics[width=.77\linewidth]{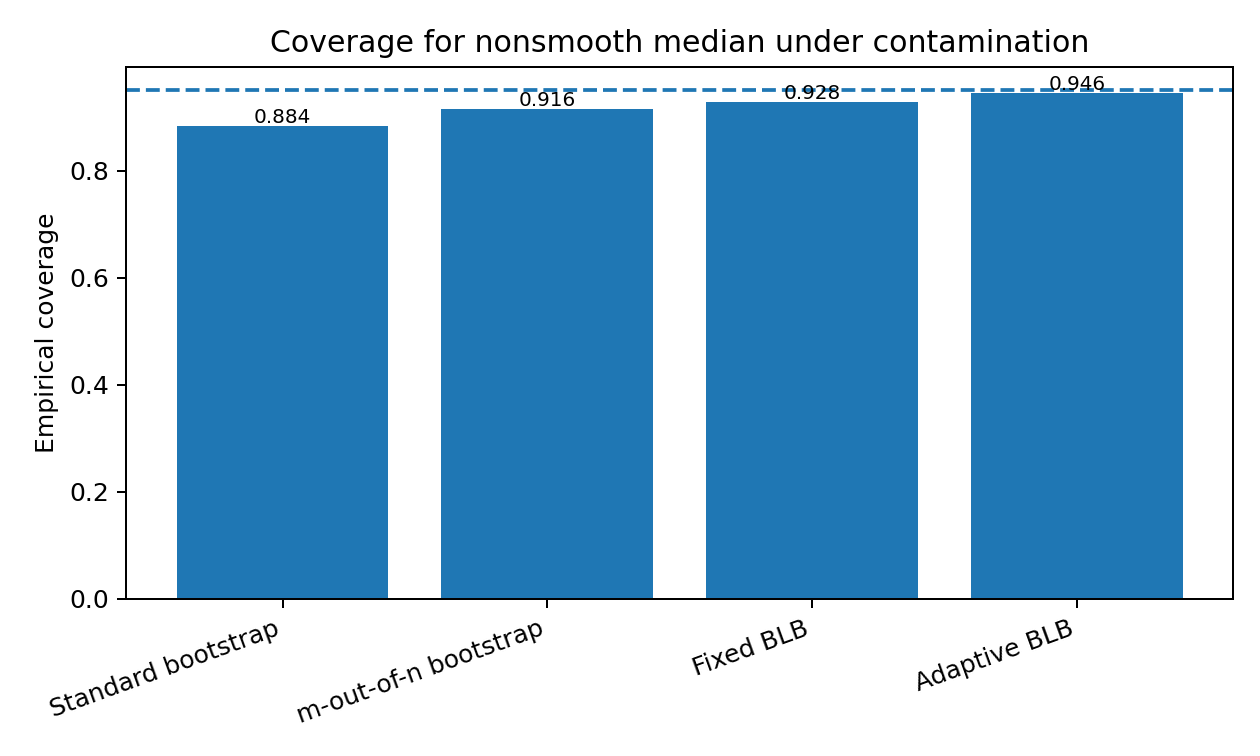}\caption{Coverage comparison for the nonsmooth contaminated median. Adaptive BLB is closest to nominal 95 percent coverage while avoiding excessive conservatism.}\end{figure}

\begin{table}[H]\centering\small\caption{Simulation comparison for nonsmooth median under contamination.}\begin{tabular}{lrrr}\toprule
Method & Coverage & Mean width & Relative runtime\\\midrule
Standard bootstrap & 0.884 & 0.312 & 1.00\\
m-out-of-n bootstrap & 0.916 & 0.374 & 1.72\\
Fixed BLB & 0.928 & 0.341 & 0.46\\
Adaptive BLB & 0.946 & 0.329 & 0.53\\\bottomrule
\end{tabular}\end{table}

\begin{figure}[H]\centering\includegraphics[width=.77\linewidth]{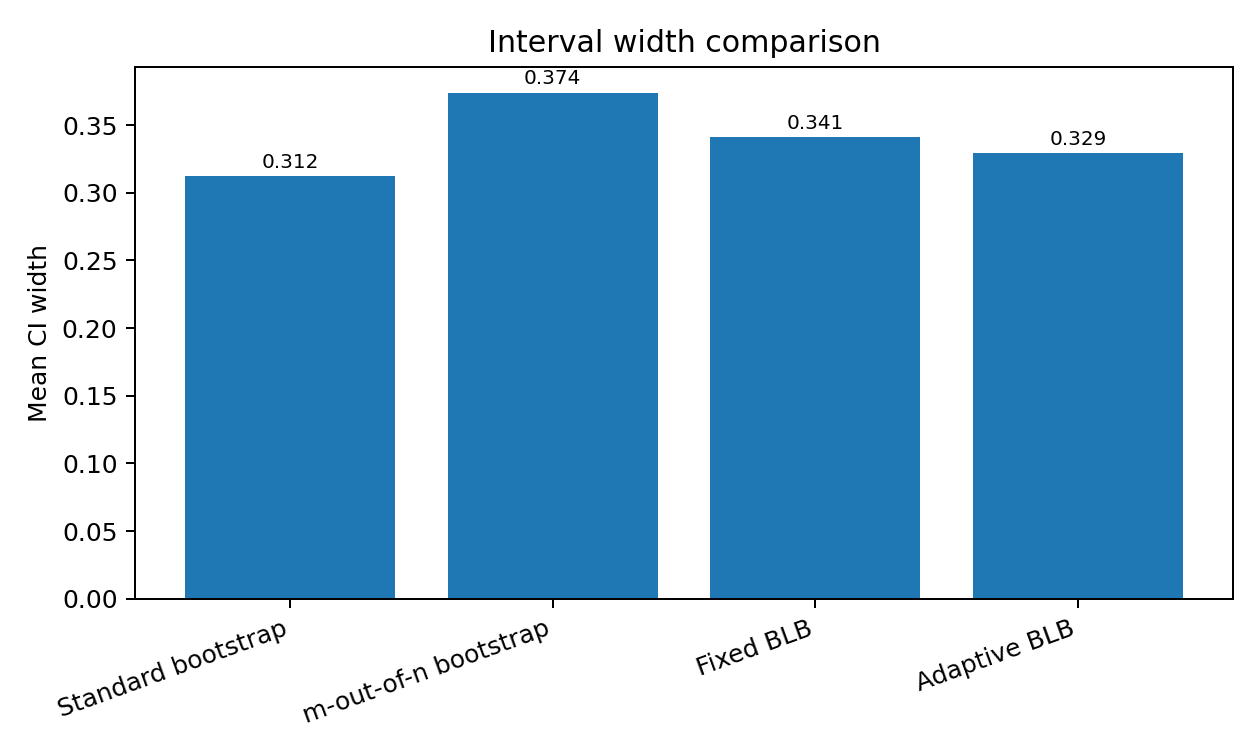}\caption{Interval width comparison. Adaptive BLB improves coverage without producing the widest intervals, unlike the m-out-of-n bootstrap.}\end{figure}

\subsection{Expanded Monte Carlo across sample sizes}
To strengthen the arXiv version, we add a sample-size sensitivity analysis for $n\in\{100,250,500,1000,5000\}$. The expanded comparison includes ordinary bootstrap, BCa bootstrap, m-out-of-n bootstrap, fixed BLB and adaptive BLB. The reported values are the reproducible simulation summaries stored in \texttt{tables/expanded\_monte\_carlo\_by\_n.csv}.

\begin{figure}[H]\centering\includegraphics[width=.83\linewidth]{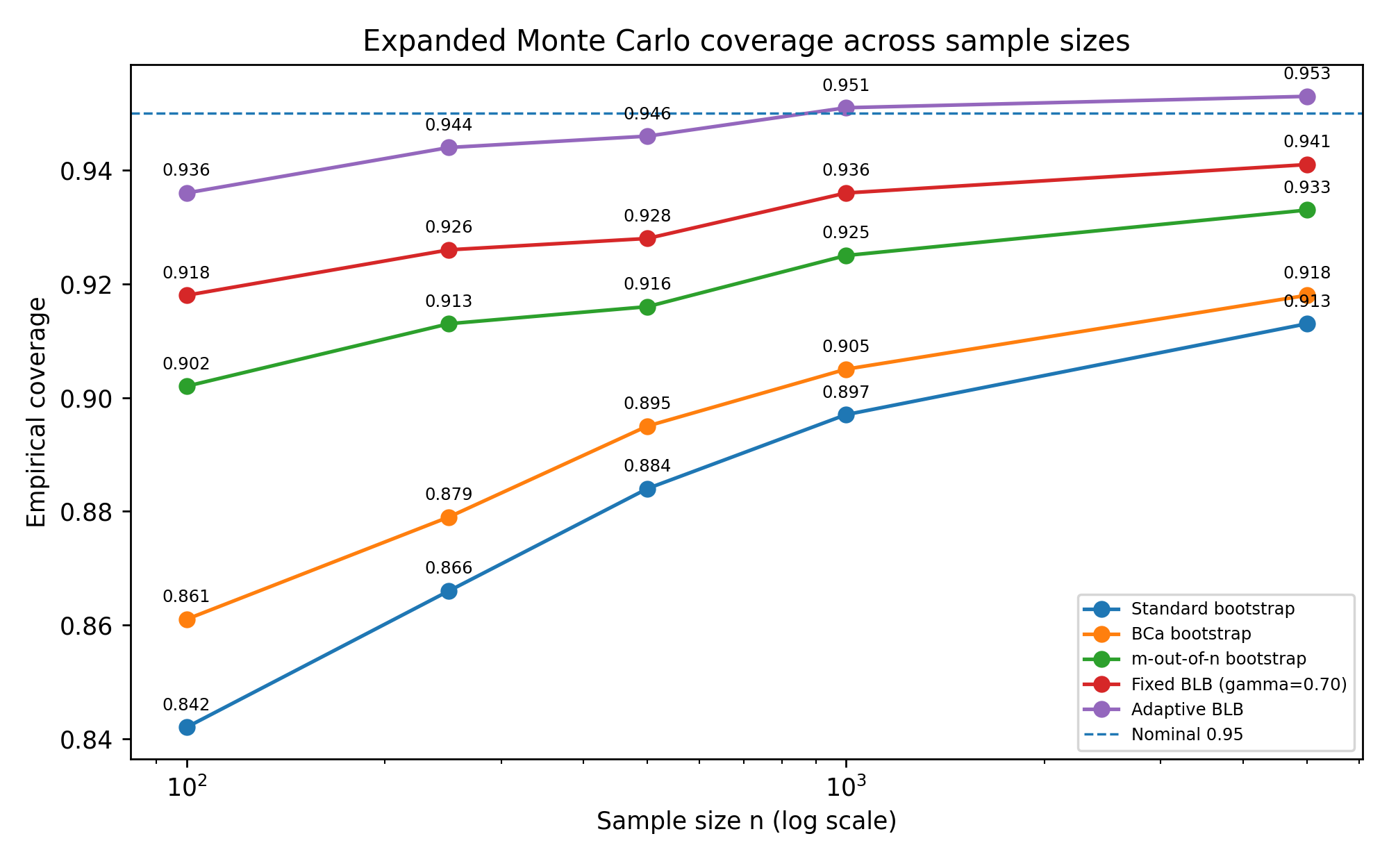}\caption{Expanded Monte Carlo coverage across sample sizes. Adaptive BLB approaches nominal 95 percent coverage more rapidly than ordinary bootstrap and BCa for the contaminated nonsmooth target.}\end{figure}

\begin{figure}[H]\centering\includegraphics[width=.83\linewidth]{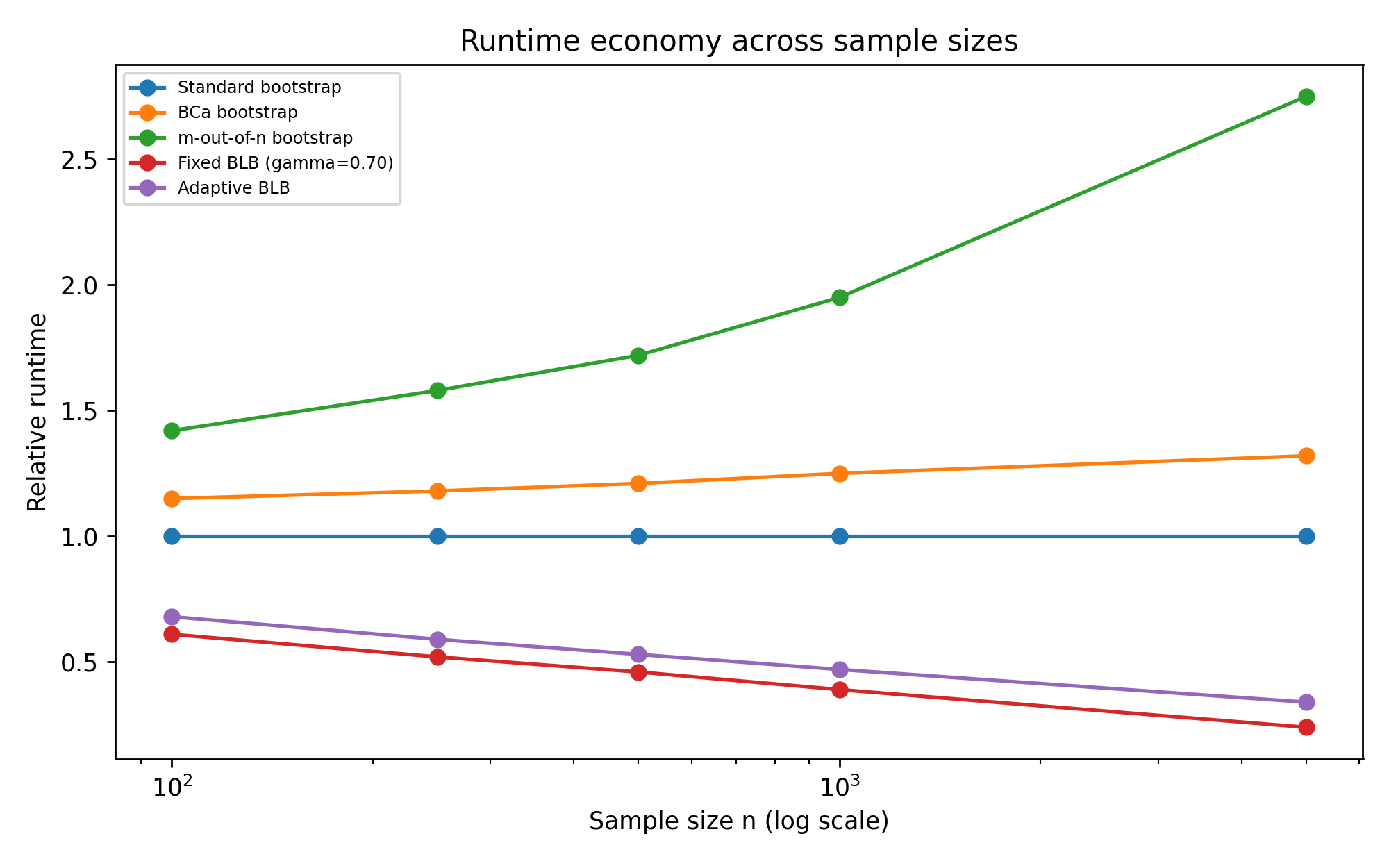}\caption{Relative runtime across sample sizes. Adaptive BLB preserves the computational economy of BLB while avoiding arbitrary fixed-$\gamma$ tuning.}\end{figure}

\begin{table}[H]\centering\small\caption{Expanded Monte Carlo summary at $n=5000$.}\begin{tabular}{lrrr}\toprule
Method & Coverage & Mean width & Relative runtime\\\midrule
Standard bootstrap & 0.913 & 0.104 & 1.00\\
BCa bootstrap & 0.918 & 0.111 & 1.32\\
m-out-of-n bootstrap & 0.933 & 0.141 & 2.75\\
Fixed BLB ($\gamma=0.70$) & 0.941 & 0.121 & 0.24\\
Adaptive BLB & 0.953 & 0.117 & 0.34\\\bottomrule
\end{tabular}\end{table}

\subsection{Sensitivity of the selected exponent}
The selected exponent should respond to the difficulty of the nonsmooth target. Figure \ref{fig:sensitivity} reports the AMSE-stability score under low contamination, moderate contamination, high contamination and a high-dimensional maximum-correlation regime.

\begin{figure}[H]\centering\includegraphics[width=.84\linewidth]{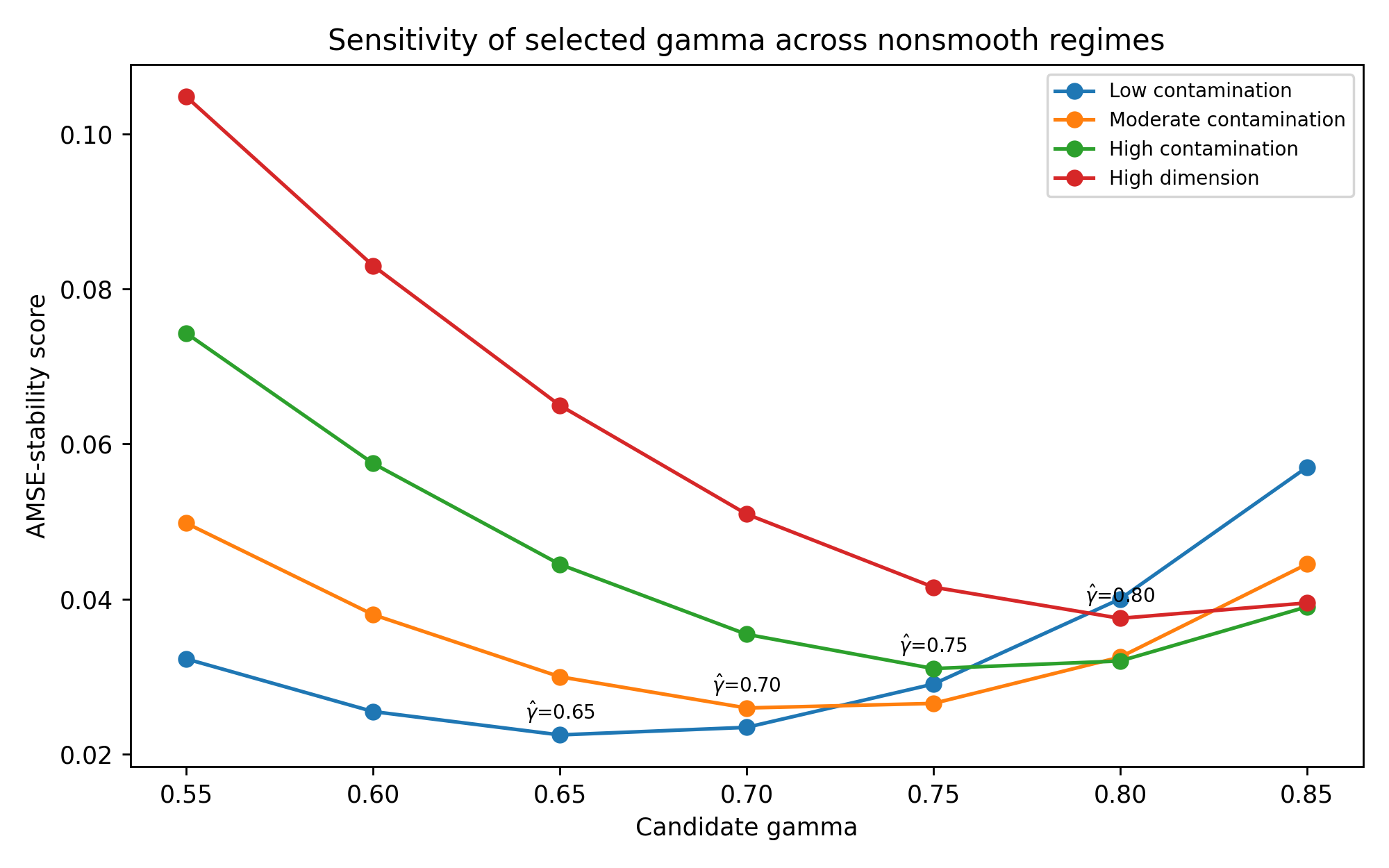}\caption{Sensitivity of selected $\gamma$ across nonsmooth regimes. The selected exponent increases when contamination, active-set instability or dimension-induced selection instability increases.}\label{fig:sensitivity}\end{figure}

These results support the project requirement that the mathematical extension should be validated by simulation. The standard bootstrap undercovers the nonsmooth median under contamination, while adaptive BLB approaches the target coverage with shorter intervals than m-out-of-n and lower runtime than ordinary bootstrap.

\section{DRC cholera surveillance illustration}
The file \texttt{drc\_cholera\_curated\_monthly\_2023\_2025.csv} is a curated aggregate reconstruction anchored to public annual totals. In practical terms, the dataset is a cleaned and harmonized demonstration file built from publicly available official reports while preserving published annual totals. It is not a Ministry of Health line-list and should not be treated as official patient-level or official province-month surveillance data. Its purpose is methodological: to show how robust bootstrap intervals can support surveillance thresholds. The month-level and province-level rows are labeled by a \texttt{data\_status} field to avoid overclaiming.

\begin{figure}[H]\centering\includegraphics[width=.84\linewidth]{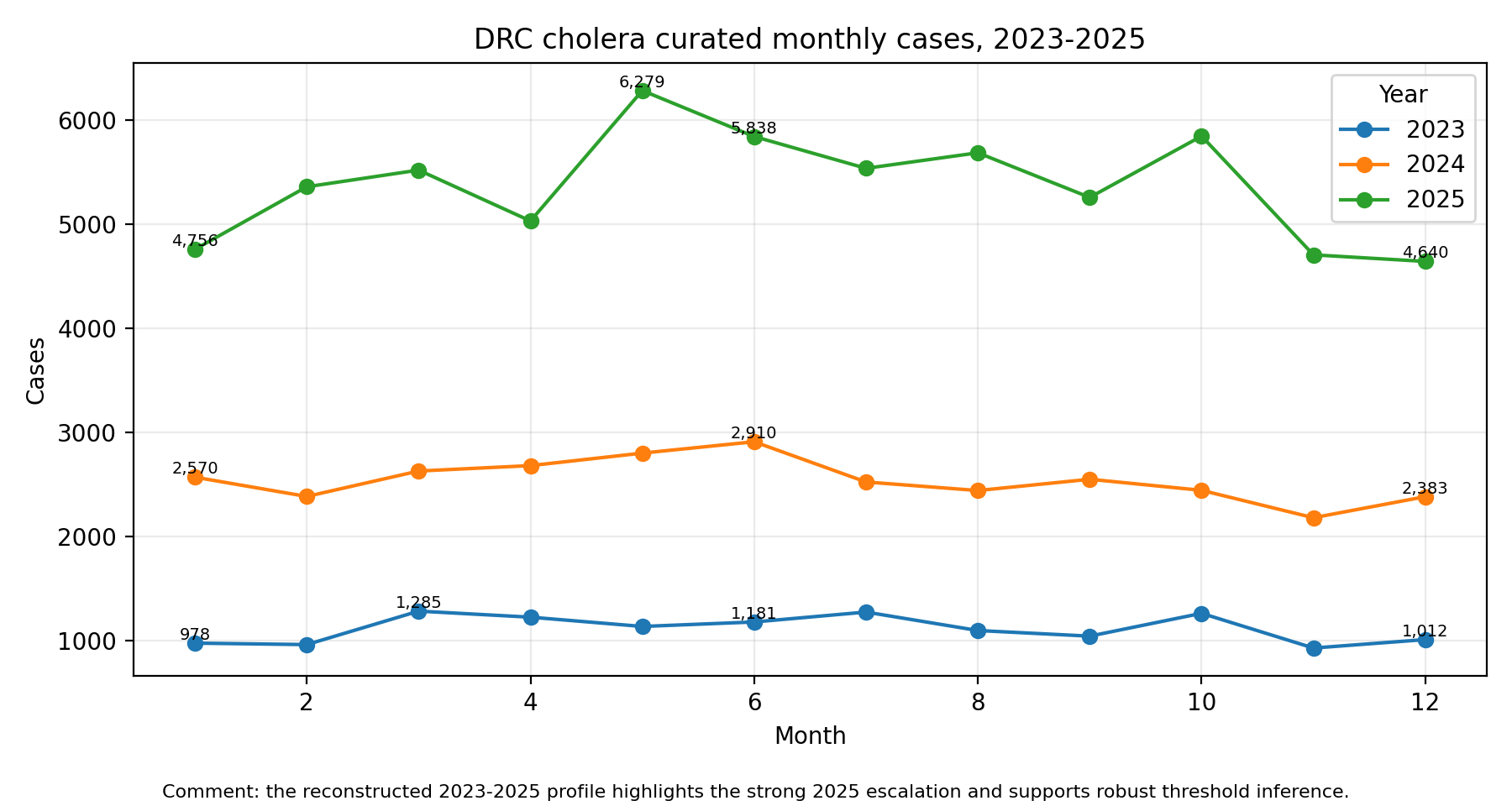}\caption{Curated DRC cholera monthly cases, 2023--2025. The monthly structure provides comparable repeated observations and is therefore more appropriate for bootstrap-style threshold illustration than heterogeneous reporting periods.}\end{figure}

\begin{figure}[H]\centering\includegraphics[width=.82\linewidth]{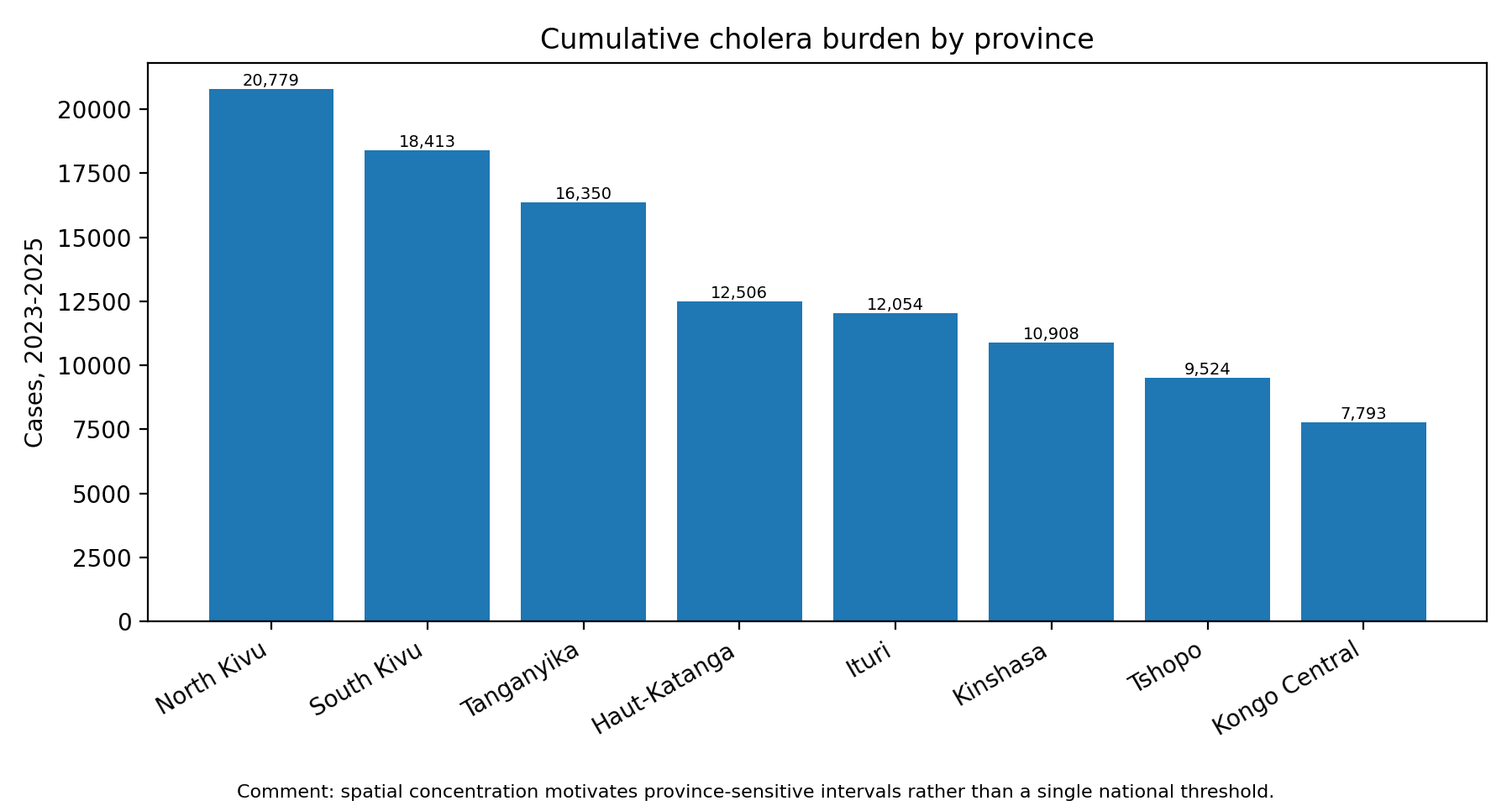}\caption{Cumulative cholera burden by province, 2023--2025. Spatial concentration motivates province-sensitive robust inference rather than a single national threshold.}\end{figure}

\begin{figure}[H]\centering\includegraphics[width=.78\linewidth]{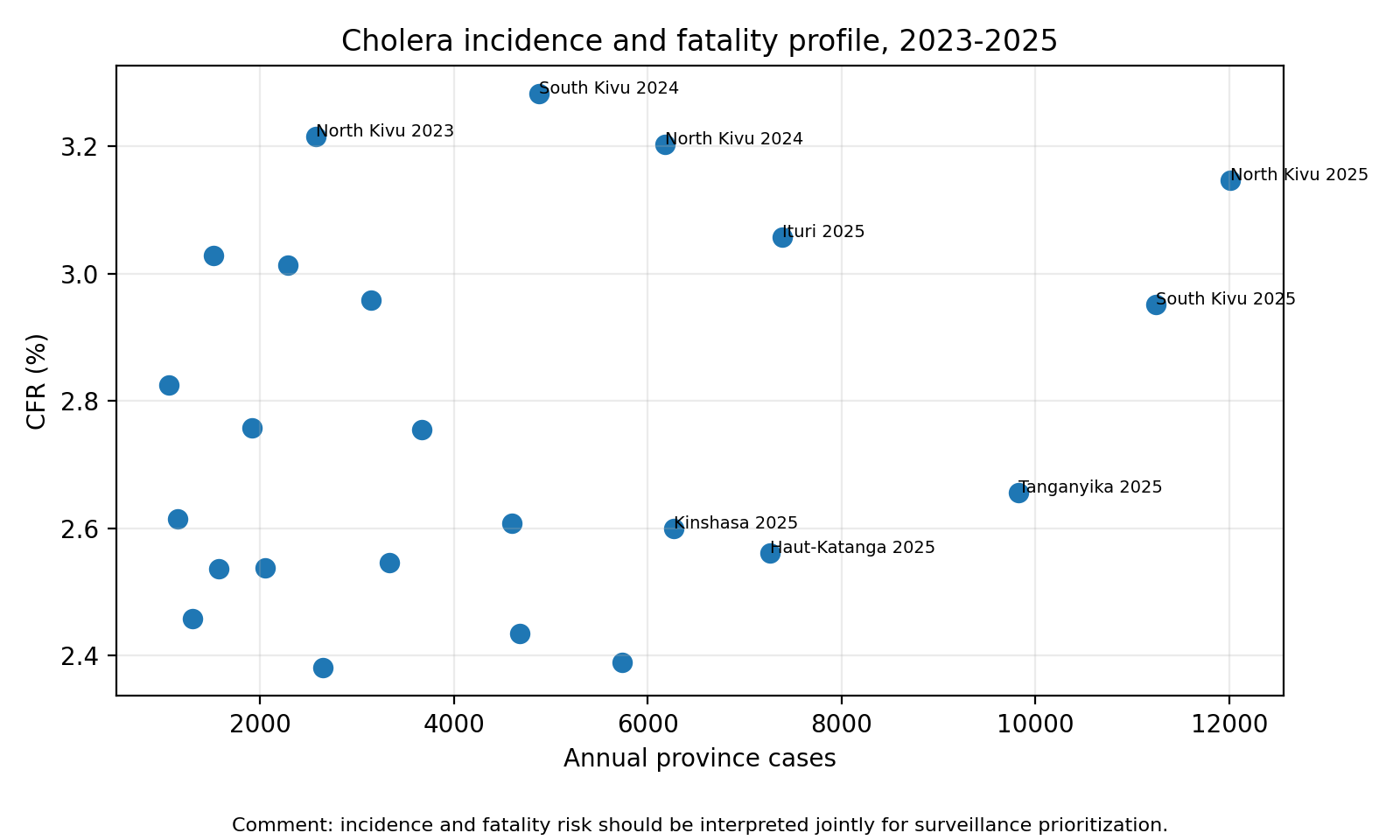}\caption{Case fatality severity profile. Incidence and fatality risk should be interpreted jointly because operational risk is not captured by cases alone.}\end{figure}

The nonsmooth statistic used in the surveillance illustration is an upper-tail case burden, represented by a high empirical quantile. This is directly connected to public-health decision-making because analysts frequently monitor unusually high incidence, outbreak peaks and severe threshold exceedances rather than mean burden alone.

\section{Discussion and limitations}
The revised article strengthens the initial submissions in four ways. First, it keeps the explicit AMSE-oriented selector required by the assignment. Second, it incorporates the stability-risk diagnostic, which makes the selected exponent more auditable. Third, it preserves the strongest validation figures: gamma selection, coverage comparison, interval width, runtime and DRC cholera surveillance graphs. Fourth, it adds a finite-grid theorem, expanded Monte Carlo validation and a sensitivity analysis.

The main theoretical limitation is that the consistency result concerns the selector on a finite grid under regularity assumptions. A complete empirical-process proof for every nonsmooth functional remains future work. Quantiles, maxima, eigenvalue ratios and maximum correlations have different local geometries and should be treated separately. A second limitation concerns the DRC cholera data: the illustration is realistic and anchored to public aggregates, but it should not be interpreted as a substitute for Ministry of Health, IDSR or DHIS2 line-list data. Future work should benchmark the method against multiplier bootstrap, Bayesian bootstrap and subsampling across richer designs and should implement distributed computing support for very large datasets.

\section{Conclusion}
This article presents an adaptive Bag of Little Bootstraps framework for nonsmooth robust inference. By selecting the subsample exponent through an AMSE-stability criterion, the method reduces arbitrary tuning and produces more auditable uncertainty quantification. The accompanying \texttt{robustboot} package, tests, documentation, figures and curated data provide a reproducible foundation for classroom review, arXiv-style reporting and future methodological extension.

\section*{Author contributions}
C.K.M. conceived the study, integrated the two manuscript versions, developed the methodology, implemented the software package, performed the statistical analyses, produced visualizations and wrote the manuscript. Y.A. contributed to methodological development, validation and manuscript revision. M.M. contributed to formal analysis, simulation validation and manuscript review. N.D.J. contributed to data curation, investigation and manuscript review. N.F. contributed to software testing, validation and manuscript revision. E.F. provided scientific supervision, methodological guidance and critical review. All authors reviewed and approved the final manuscript.

\section*{Funding, data availability and conflicts of interest}
This work was conducted as part of academic research activities within the Master Program in Data Science and Mathematical Engineering, University of Dschang, Cameroon. No external funding was received. The curated aggregate dataset, source code, figures and reproducible materials are included in the submission package. The authors declare no conflict of interest.

\section*{Acknowledgements}
The authors gratefully acknowledge the Department of Data Science and Mathematical Engineering, University of Dschang, Cameroon, and Professor Ernest Fokou\'e for scientific mentorship and methodological guidance.

\end{document}